\documentclass[12pt,oneside]{amsart}

\usepackage{amssymb}
\usepackage{amsfonts}
\usepackage{amscd}
\usepackage[matrix, arrow, curve]{xy}
\usepackage{graphicx}

\numberwithin{equation}{section}
\newtheorem{theorem}{Theorem}[section]

\newtheorem{proposition}[theorem]{Proposition}

\newtheorem{lemma}[theorem]{Lemma}
\theoremstyle{definition}
\newtheorem{definition}[theorem]{Definition}
\theoremstyle{remark}
\newtheorem*{remark}{Remark}
\newtheorem{example}{Example}

\begin{document}
\newcommand{\M}{\mathcal{M}}
\newcommand{\F}{\mathcal{F}}

\newcommand{\Teich}{\mathcal{T}_{g,N+1}^{(1)}}
\newcommand{\T}{\mathrm{T}}
\newcommand{\corr}{\bf}
\newcommand{\vac}{|0\rangle}
\newcommand{\Ga}{\Gamma}
\newcommand{\new}{\bf}
\newcommand{\define}{\def}
\newcommand{\redefine}{\def}
\newcommand{\Cal}[1]{\mathcal{#1}}
\renewcommand{\frak}[1]{\mathfrak{{#1}}}
\newcommand{\Hom}{\rm{Hom}\,}
\newcommand{\refE}[1]{(\ref{E:#1})}
\newcommand{\refCh}[1]{Chapter~\ref{Ch:#1}}
\newcommand{\refS}[1]{Section~\ref{S:#1}}
\newcommand{\refSS}[1]{Section~\ref{SS:#1}}
\newcommand{\refT}[1]{Theorem~\ref{T:#1}}
\newcommand{\refO}[1]{Observation~\ref{O:#1}}
\newcommand{\refP}[1]{Proposition~\ref{P:#1}}
\newcommand{\refD}[1]{Definition~\ref{D:#1}}
\newcommand{\refC}[1]{Corollary~\ref{C:#1}}
\newcommand{\refL}[1]{Lemma~\ref{L:#1}}
\newcommand{\refEx}[1]{Example~\ref{Ex:#1}}
\newcommand{\R}{\ensuremath{\mathbb{R}}}
\newcommand{\C}{\ensuremath{\mathbb{C}}}
\newcommand{\N}{\ensuremath{\mathbb{N}}}
\newcommand{\Q}{\ensuremath{\mathbb{Q}}}
\renewcommand{\P}{\ensuremath{\mathcal{P}}}
\newcommand{\Z}{\ensuremath{\mathbb{Z}}}
\newcommand{\kv}{{k^{\vee}}}
\renewcommand{\l}{\lambda}
\newcommand{\gb}{\overline{\mathfrak{g}}}
\newcommand{\dt}{\tilde d}     
\newcommand{\hb}{\overline{\mathfrak{h}}}
\newcommand{\g}{\mathfrak{g}}
\newcommand{\h}{\mathfrak{h}}
\newcommand{\gh}{\widehat{\mathfrak{g}}}
\newcommand{\ghN}{\widehat{\mathfrak{g}_{(N)}}}
\newcommand{\gbN}{\overline{\mathfrak{g}_{(N)}}}
\newcommand{\tr}{\mathrm{tr}}
\newcommand{\gln}{\mathfrak{gl}(n)}
\newcommand{\son}{\mathfrak{so}(n)}
\newcommand{\spnn}{\mathfrak{sp}(2n)}
\newcommand{\sln}{\mathfrak{sl}}
\newcommand{\sn}{\mathfrak{s}}
\newcommand{\so}{\mathfrak{so}}
\newcommand{\spn}{\mathfrak{sp}}
\newcommand{\tsp}{\mathfrak{tsp}(2n)}
\newcommand{\gl}{\mathfrak{gl}}
\newcommand{\slnb}{{\overline{\mathfrak{sl}}}}
\newcommand{\snb}{{\overline{\mathfrak{s}}}}
\newcommand{\sob}{{\overline{\mathfrak{so}}}}
\newcommand{\spnb}{{\overline{\mathfrak{sp}}}}
\newcommand{\glb}{{\overline{\mathfrak{gl}}}}
\newcommand{\Hwft}{\mathcal{H}_{F,\tau}}
\newcommand{\Hwftm}{\mathcal{H}_{F,\tau}^{(m)}}

\newcommand{\car}{{\mathfrak{h}}}    
\newcommand{\bor}{{\mathfrak{b}}}    
\newcommand{\nil}{{\mathfrak{n}}}    
\newcommand{\vp}{{\varphi}}
\newcommand{\bh}{\widehat{\mathfrak{b}}}  
\newcommand{\bb}{\overline{\mathfrak{b}}}  
\newcommand{\Vh}{\widehat{\mathcal V}}
\newcommand{\KZ}{Kniz\-hnik-Zamo\-lod\-chi\-kov}
\newcommand{\TUY}{Tsuchia, Ueno  and Yamada}
\newcommand{\KN} {Kri\-che\-ver-Novi\-kov}
\newcommand{\pN}{\ensuremath{(P_1,P_2,\ldots,P_N)}}
\newcommand{\xN}{\ensuremath{(\xi_1,\xi_2,\ldots,\xi_N)}}
\newcommand{\lN}{\ensuremath{(\lambda_1,\lambda_2,\ldots,\lambda_N)}}
\newcommand{\iN}{\ensuremath{1,\ldots, N}}
\newcommand{\iNf}{\ensuremath{1,\ldots, N,\infty}}

\newcommand{\tb}{\tilde \beta}
\newcommand{\tk}{\tilde \varkappa}
\newcommand{\ka}{\kappa}
\renewcommand{\k}{\varkappa}
\newcommand{\ce}{{c}}

\newcommand{\Pif} {P_{\infty}}
\newcommand{\Pinf} {P_{\infty}}
\newcommand{\PN}{\ensuremath{\{P_1,P_2,\ldots,P_N\}}}
\newcommand{\PNi}{\ensuremath{\{P_1,P_2,\ldots,P_N,P_\infty\}}}
\newcommand{\Fln}[1][n]{F_{#1}^\lambda}
\newcommand{\tang}{\mathrm{T}}
\newcommand{\Kl}[1][\lambda]{\can^{#1}}
\newcommand{\A}{\mathcal{A}}
\newcommand{\U}{\mathcal{U}}
\newcommand{\V}{\mathcal{V}}
\newcommand{\W}{\mathcal{W}}
\renewcommand{\O}{\mathcal{O}}
\newcommand{\Ae}{\widehat{\mathcal{A}}}
\newcommand{\Ah}{\widehat{\mathcal{A}}}
\newcommand{\La}{\mathcal{L}}
\newcommand{\Le}{\widehat{\mathcal{L}}}
\newcommand{\Lh}{\widehat{\mathcal{L}}}
\newcommand{\eh}{\widehat{e}}
\newcommand{\Da}{\mathcal{D}}
\newcommand{\kndual}[2]{\langle #1,#2\rangle}
\newcommand{\cins}{\frac 1{2\pi\mathrm{i}}\int_{C_S}}
\newcommand{\cinsl}{\frac 1{24\pi\mathrm{i}}\int_{C_S}}
\newcommand{\cinc}[1]{\frac 1{2\pi\mathrm{i}}\int_{#1}}
\newcommand{\cintl}[1]{\frac 1{24\pi\mathrm{i}}\int_{#1 }}
\newcommand{\w}{\omega}
\newcommand{\ord}{\operatorname{ord}}
\newcommand{\res}{\operatorname{res}}
\newcommand{\nord}[1]{:\mkern-5mu{#1}\mkern-5mu:}
\newcommand{\codim}{\operatorname{codim}}
\newcommand{\ad}{\operatorname{ad}}
\newcommand{\Ad}{\operatorname{Ad}}
\newcommand{\supp}{\operatorname{supp}}

\newcommand{\Fn}[1][\lambda]{\mathcal{F}^{#1}}
\newcommand{\Fl}[1][\lambda]{\mathcal{F}^{#1}}
\renewcommand{\Re}{\mathrm{Re}}

\newcommand{\ha}{H^\alpha}

\define\ldot{\hskip 1pt.\hskip 1pt}
\define\ifft{\qquad\text{if and only if}\qquad}
\define\a{\alpha}
\redefine\d{\delta}
\define\w{\omega}
\define\ep{\epsilon}
\redefine\b{\beta} \redefine\t{\tau} \redefine\i{{\,\mathrm{i}}\,}
\define\ga{\gamma}
\define\cint #1{\frac 1{2\pi\i}\int_{C_{#1}}}
\define\cintta{\frac 1{2\pi\i}\int_{C_{\tau}}}
\define\cintt{\frac 1{2\pi\i}\oint_{C}}
\define\cinttp{\frac 1{2\pi\i}\int_{C_{\tau'}}}
\define\cinto{\frac 1{2\pi\i}\int_{C_{0}}}
\define\cinttt{\frac 1{24\pi\i}\int_C}
\define\cintd{\frac 1{(2\pi \i)^2}\iint\limits_{C_{\tau}\,C_{\tau'}}}
\define\dintd{\frac 1{(2\pi \i)^2}\iint\limits_{C\,C'}}
\define\cintdr{\frac 1{(2\pi \i)^3}\int_{C_{\tau}}\int_{C_{\tau'}}
\int_{C_{\tau''}}}
\define\im{\operatorname{Im}}
\define\re{\operatorname{Re}}
\define\res{\operatorname{res}}
\redefine\deg{\operatornamewithlimits{deg}}
\define\ord{\operatorname{ord}}
\define\rank{\operatorname{rank}}
\define\fpz{\frac {d }{dz}}
\define\dzl{\,{dz}^\l}
\define\pfz#1{\frac {d#1}{dz}}

\define\K{\Cal K}
\define\U{\Cal U}
\redefine\O{\Cal O}
\define\He{\text{\rm H}^1}
\redefine\H{{\mathrm{H}}}
\define\Ho{\text{\rm H}^0}
\define\A{\Cal A}
\define\Do{\Cal D^{1}}
\define\Dh{\widehat{\mathcal{D}}^{1}}
\redefine\L{\Cal L}
\newcommand{\ND}{\ensuremath{\mathcal{N}^D}}
\redefine\D{\Cal D^{1}}
\define\KN {Kri\-che\-ver-Novi\-kov}
\define\Pif {{P_{\infty}}}
\define\Uif {{U_{\infty}}}
\define\Uifs {{U_{\infty}^*}}
\define\KM {Kac-Moody}
\define\Fln{\Cal F^\lambda_n}
\define\gb{\overline{\mathfrak{ g}}}
\define\G{\overline{\mathfrak{ g}}}
\define\Gb{\overline{\mathfrak{ g}}}
\redefine\g{\mathfrak{ g}}
\define\Gh{\widehat{\mathfrak{ g}}}
\define\gh{\widehat{\mathfrak{ g}}}
\define\Ah{\widehat{\Cal A}}
\define\Lh{\widehat{\Cal L}}
\define\Ugh{\Cal U(\Gh)}
\define\Xh{\hat X}
\define\Tld{...}
\define\iN{i=1,\ldots,N}
\define\iNi{i=1,\ldots,N,\infty}
\define\pN{p=1,\ldots,N}
\define\pNi{p=1,\ldots,N,\infty}
\define\de{\delta}

\define\kndual#1#2{\langle #1,#2\rangle}
\define \nord #1{:\mkern-5mu{#1}\mkern-5mu:}
\newcommand{\MgN}{\mathcal{M}_{g,N}} 
\newcommand{\MgNeki}{\mathcal{M}_{g,N+1}^{(k,\infty)}} 
\newcommand{\MgNeei}{\mathcal{M}_{g,N+1}^{(1,\infty)}} 
\newcommand{\MgNekp}{\mathcal{M}_{g,N+1}^{(k,p)}} 
\newcommand{\MgNkp}{\mathcal{M}_{g,N}^{(k,p)}} 
\newcommand{\MgNk}{\mathcal{M}_{g,N}^{(k)}} 
\newcommand{\MgNekpp}{\mathcal{M}_{g,N+1}^{(k,p')}} 
\newcommand{\MgNekkpp}{\mathcal{M}_{g,N+1}^{(k',p')}} 
\newcommand{\MgNezp}{\mathcal{M}_{g,N+1}^{(0,p)}} 
\newcommand{\MgNeep}{\mathcal{M}_{g,N+1}^{(1,p)}} 
\newcommand{\MgNeee}{\mathcal{M}_{g,N+1}^{(1,1)}} 
\newcommand{\MgNeez}{\mathcal{M}_{g,N+1}^{(1,0)}} 
\newcommand{\MgNezz}{\mathcal{M}_{g,N+1}^{(0,0)}} 
\newcommand{\MgNi}{\mathcal{M}_{g,N}^{\infty}} 
\newcommand{\MgNe}{\mathcal{M}_{g,N+1}} 
\newcommand{\MgNep}{\mathcal{M}_{g,N+1}^{(1)}} 
\newcommand{\MgNp}{\mathcal{M}_{g,N}^{(1)}} 
\newcommand{\Mgep}{\mathcal{M}_{g,1}^{(p)}} 
\newcommand{\MegN}{\mathcal{M}_{g,N+1}^{(1)}} 

\define \sinf{{\widehat{\sigma}}_\infty}
\define\Wt{\widetilde{W}}
\define\St{\widetilde{S}}
\newcommand{\SigmaT}{\widetilde{\Sigma}}
\newcommand{\hT}{\widetilde{\frak h}}
\define\Wn{W^{(1)}}
\define\Wtn{\widetilde{W}^{(1)}}
\define\btn{\tilde b^{(1)}}
\define\bt{\tilde b}
\define\bn{b^{(1)}}
\define \ainf{{\frak a}_\infty} 

%
\define\eps{\varepsilon}    
\newcommand{\e}{\varepsilon}
\define\doint{({\frac 1{2\pi\i}})^2\oint\limits _{C_0}
       \oint\limits _{C_0}}                            
\define\noint{ {\frac 1{2\pi\i}} \oint}   
\define \fh{{\frak h}}     
\define \fg{{\frak g}}     
\define \GKN{{\Cal G}}   
\define \gaff{{\hat\frak g}}   
\define\V{\Cal V}
\define \ms{{\Cal M}_{g,N}} 
\define \mse{{\Cal M}_{g,N+1}} 
\define \tOmega{\Tilde\Omega}
\define \tw{\Tilde\omega}
\define \hw{\hat\omega}
\define \s{\sigma}
\define \car{{\frak h}}    
\define \bor{{\frak b}}    
\define \nil{{\frak n}}    
\define \vp{{\varphi}}
\define\bh{\widehat{\frak b}}  
\define\bb{\overline{\frak b}}  
\define\KZ{Knizhnik-Zamolodchikov}
\define\ai{{\alpha(i)}}
\define\ak{{\alpha(k)}}
\define\aj{{\alpha(j)}}
\newcommand{\calF}{{\mathcal F}}
\newcommand{\ferm}{{\mathcal F}^{\infty /2}}
\newcommand{\Aut}{\operatorname{Aut}}
\newcommand{\End}{\operatorname{End}}
\newcommand{\laxgl}{\overline{\mathfrak{gl}}}
\newcommand{\laxsl}{\overline{\mathfrak{sl}}}
\newcommand{\laxso}{\overline{\mathfrak{so}}}
\newcommand{\laxsp}{\overline{\mathfrak{sp}}}
\newcommand{\laxs}{\overline{\mathfrak{s}}}
\newcommand{\laxg}{\overline{\frak g}}
\newcommand{\bgl}{\laxgl(n)}
\newcommand{\tX}{\widetilde{X}}
\newcommand{\tY}{\widetilde{Y}}
\newcommand{\tZ}{\widetilde{Z}}


\title[]{Spectral curves of the hyperelliptic Hitchin systems}
\author[O.K.Sheinman]{O.K.Sheinman}
\maketitle
\begin{abstract}
A description of the class of spectral curves, and explicit formulas for algebraic-geometric action-angle coordinates are obtained for the Hitchin systems on hyperelliptic curves, for any complex simple Lie algebra of the types $A_l$, $B_l$, $C_l$.
\end{abstract}
\tableofcontents
\section{Introduction}
The aim of this paper is a further effectivization of the separation of variables scheme for Hitchin systems, based on their Lax representation originally proposed in \cite{Kr_Lax}, and further developed in \cite{Sh_DGr,Shein_UMN2016} (see also references therein).

In the present work we pursue nearly the same objective that are pursued in \cite{Gaw,Previato}. In particular, in \cite{Gaw} the formulae for the action--angle coordinates, and $\theta$-formulae for solutions are obtained for the rank 2 genus 2 Hitchin systems.

The geometry of separation of variables for Hitchin systems has been studied in \cite{Hur,GNR}. Let $\M_{r,d}$ be the moduli space of holomorphic rank $r$ degree $d$ vector bundles on  a Riemann surface $\Sigma$, $h=\dim \M_{r,d}$. Then the following holds:
\begin{theorem}[\cite{GNR}]\label{T:GNR}
There exist a birational map
\[
   \varphi :\, T^*\M_{r,d}\to (T^*\Sigma)^{[h]}
\]
which is a symplectomorphism on open dense subsets.
\end{theorem}
\noindent Here $(T^*\Sigma)^{[h]}$ denotes the Hilbert scheme of $h$ points on $T^*\Sigma$ whose open dense subset coincides with the $h$th symmetric power of $T^*\Sigma$ with cut out diagonals.

The algebra of separation of variables, at least its part we need here, is presented in \cite{BT2,DT}. From a more general point of view it goes back to \cite{BT1,DKN,Skl}.

The relation between the separation of variables and the algebro-geometric version of the inverse scattering method for the Hitchin systems of type ${\rm A}_n$ (and for a certain wider class of systems), apart of many other important results, is established in \cite{Kr_Lax}. It is shown there that the canonical variables of the symplectic structure on the right hand side of the symplectimorphism in \refT{GNR} are nothing but the poles of the Baker--Akhiezer function, and their conjugated variables.

As it is shown in \cite{Kr_Lax}, coefficients of spectral curves of Hitchin systems form spaces of meromorphic functions on $\Sigma$ whose divisors are multiples of a certain canonical divisor. After that, in accordance with definition of the Hitchin systems, in order to obtain the Hamiltonians it only remains to choose appropriate bases in those spaces. This is what we begin with in this paper. As the result, we obtain an explicit description of the class of spectral curves of Hitchin systems on hyperelliptic curves of arbitrary genera, and for arbitrary complex semisimple Lie algebras (\refP{f_base} of \refSS{dx/y}). Before, such description was known for Calogero--Moser systems \cite{DW,PD'Hock} (it was obtained by a different technique).

We proceed then with the description of the holomorphic differentials on spectral curves (\refP{h_diff}). All together enables us to carry out the scheme of separation of variables and explicitly write down the algebraic-geometric action-angle coordinates. To give a taste of the answer, we come up with a simplest example here. For the rank~2 Hitchin system on a genus~2 hyperelliptic curve $y^2=P_5(x)(=x^5+\ldots)\ $, the spectral curve is a full intersection of the surfaces $\l^2=H_0+H_1x+H_2x^2$ and $y^2=P_5(x)$, the phase space is given by the triples $\ga_i=(x_i,y_i,z_i)$ ($i=1,2,3$) (where $x_i$, $y_i$, $z_i$ satisfy the above relations) with the symplectic form $d\l_1\wedge\frac{dx_1}{y_1}+d\l_2\wedge\frac{dx_2}{y_2}+d\l_3\wedge\frac{dx_3}{y_3}$, the action coordinates coincide with $H_0$, $H_1$, $H_2$, they form a solution to the system of linear equations
\begin{equation}\label{E:action_rk(2)}
     \l_i^2=H_0+H_1x_i+H_2x_i^2\ (i=1,2,3),
\end{equation}
and the angle coordinates are given by the relations
\begin{equation}\label{E:angle_rk(2)}
     \phi_k=\sum_{i=1}^3 \int^{\ga_i}\frac{x^kdx}{\l y}\  (k=0,1,2).
\end{equation}

We derive these results from the analitic properties of the Lax operator for Hitchin systems in the Tyurin parametrization given in \cite{Kr_Lax,Sh_DGr,Shein_UMN2016}. Observe that there is no explicit expression for it. However, the knowledge of the analitic properties is sufficient to obtain the above formulated explicit results. Observe also that the important explicit Lax representation for the Calogero--Moser systems \cite{Kr_FAN_78}, its generalizations in \cite{PD'Hock} and in \cite{Sh_DGr}, and also the example on a hyperelliptic curve \cite{Kr_Lax}, indeed are the Lax representations for Hitchin systems with additional marked points.

Many results of the present paper hold for an arbitrary complex semisimple Lie algebra $\g$, in particular, for $\g$ of type $D_l$ or $G_2$ but there are certain peculiarities, see our remarks on $D_l$ in Sections \ref{SS:num_int},  \ref{SS:dx/y}.

I am grateful to S.P.Novikov who has drown my attention to the problem, and to D.V.Talalaev for his interest and discussions.

\section{Description of spectral curves}\label{S:descr}
In this section we propose a description of the class of spectral curves of Hitchin systems on hyperelliptic curves, with arbitrary semisimple Lie algebras.
\subsection{Lax operator of a Hitchin system}\label{SS:Lax}
We begin with a general definition of Lax operators with a spectral parameter on a Riemann surface following \cite{Shein_UMN2016} (the theory of the corresponding integrable systems was pioneered by I.Krichever in \cite{Kr_Lax}, see \refEx{Ex1} below).

Let $\g$ be a semi-simple Lie algebra over $\C$, $\h$ be its Cartan subalgebra, and $h\in\h$ be such element that  $p_i=\a_i(h)\in\Z_+$ for every simple root $\a_i$ of $\g$. Let $\g_p=\{ X\in\g\ |\ (\ad h)X=pX \}$, and $k=\max\{p\ |\ \g_p\ne 0\}$. Then the decomposition $\g=\bigoplus\limits_{i=-k}^{k}\g_p$ gives a $\Z$-grading on $\g$. We call $k$ a \emph{depth} of the grading. Obviously, $\g_p=\bigoplus\limits_{\substack{\a\in R\\ \a(h)=p}}\g_\a$ where $R$ is the root system of $\g$. Define also the following filtration on $\g$:  $\tilde\g_p=\bigoplus\limits_{q=-k}^p\g_q$. Then $\tilde\g_p\subset\tilde\g_{p+1}$ ($p\ge -k$), $\tilde\g_{-k}=\g_{-k},\ldots,\tilde\g_k=\g$, $\tilde\g_p=\g$, $p>k$.

Let $\Sigma$ be a complex compact Riemann surface with two given finite sets of marked points: $\Pi$ and $\Gamma$. We fix a non-negative divisor $D$ on $\Sigma$ supported on $\Pi$.
\begin{definition}\label{D:Lax_gen}
By \emph{Lax operator} $L$ we mean a meromorphic mapping $\Sigma\to\g$  such that $(L)+k\sum_{\ga\in\Gamma}\ga+D\ge 0$, and $L$ has the Laurent decomposition of the following form at the points in $\Gamma$:
\begin{equation}\label{E:ga_expan}
   L(z)=\sum\limits_{p=-k}^\infty L_pz^p,\ L_p\in\tilde\g_p
\end{equation}
where $z$ is a local coordinate in the neighborhood of a $\ga\in\Gamma$.
\end{definition}
In general, the grading element $h$ may vary from one to another point of $\Gamma$ ($h=h_\ga$). For simplicity, we assume that $k$ is the same all over $\Gamma$, though it would be no difference otherwise.

In what follows we regard to $\ga\in\Gamma$ and to $\{ h_\ga\ |\ga\in\Gamma\}$ as to the dynamical variables of the corresponding integrable system.
\begin{example}\label{Ex:Ex1}
For $\g=\sln(n)$ an equivalent definition of the Lax operator (indeed, it was the original definition \cite{Kr_Lax}) is obtained by replacing the relation \refE{ga_expan} with the following:
\begin{equation}
 L(z)=\frac{\a_\ga\b_\ga^t}{z}+L_{0,\ga}+O(z)
\end{equation}
where $\a_\ga,\b_\ga\in\C^n$, $\b^t\a=0$ and there exist $\k_\ga\in\C$ such that $L_{0,\ga}\a_\ga=\k_\ga\a_\ga$  for every $\ga\in\Gamma$.  The parameters $\ga$, $\a_\ga$ ($\ga\in\Gamma$) are called \emph{Tyurin parameters}. They, and dual $\k_\ga$, $\b_\ga$ ($\ga\in\Gamma$) are dynamical variables of the system. For the proof of equivalence of the two above definitions, and for the Tyurin parametrization in the case $\g$ is a classical complex Lie algebra, or $\g=G_2$, we refer to \cite{Shein_UMN2016,Sh_DGr}.
\end{example}
In order to obtain the Lax operator of a Hitchin system, we must set $D$ to be equal to the divisor of a holomorphic differential \cite{Kr_Lax}. The last is denoted by $\varpi$ below: $D=(\varpi)$. In particular, $D\in\K$ where $\K$ is the canonical class.

From now on $\Sigma$ is a hyperelliptic curve given by
\begin{equation}\label{E:curve}
   y^2=x^{2g+1}+\sum_{i=0}^{2g} a_ix^i\, (=P_{2g+1}(x)).
\end{equation}
Up to the end of \refS{separ} we fix $D=2(g-1)\cdot\infty$ which is the divisor of the holomorphic differential $\varpi=\frac{dx}{y}$.

\subsection{Basis spectral invariants and their analytic properties}\label{SS:num_int}
By a \emph{basis spectral invariant} we mean a function of the form  $p_i(z)=\chi_i(L(z))$ where $\chi_i$ ($i=1,\ldots,l$), $l=\rank\g$, are basis invariant polynomials of the Lie algebra $\g$, $L$ is the Lax operator. The $d_i=\deg \chi_i$ is refered to as the degree of the basis spectral invariant $p_i$. By definition, the Hamiltonians of the system are obtained as the coefficients of expansions of basis spectral invariants in the linear combination of appropriate basis functions in the space of invariants of a given degree. Such a base is exhibited in \refSS{dx/y}.

For the Lie algebras of classic series $A_n$, $B_n$, $C_n$ the basis spectral invariants coincide with non-vanishing coefficients of the corresponding spectral curve. For the series $D_n$ it is true for all coefficients of the spectral curve except for one of them, namely, except for determinant of $L(z)$ which is equal to the square of the basis invariant given by the Pfaffian.

We summarize the analytic properties of the basis spectral invariants in the following statement.
\begin{proposition}[\cite{Kr_Lax,Shein_UMN2016,Sh_DGr}]
$1^\circ$. The basis (hence, \emph{all}) spectral invariants are holomorphic at the points in $\Gamma$.
\newline
$2^\circ$. The following holds for the basis spectral invariant $p_i$ ($i=1,\ldots,l$):
\[
   (p_i)+d_iD\ge 0.
\]
$3^\circ$. For $D\in\K$ the dimension of the space generated by basis spectral invariants is equal to $\dim\g\cdot(g-1)$.
\end{proposition}
\begin{proof}
Here, we briefly outline the proof of the theorem (for the detailed proof we refer to the quoted works).

The first statement is implied by the following fact: those poles of $L$ which belong to $\Gamma$ can be eliminated by a conjugation of $L$ by a certain local holomorphic function taking values in the corresponding group \cite{Kr_Lax}, \cite[лемма 4.2]{Shein_UMN2016}.

By $1^\circ$, the divisor of the function $p_i$ is supported at the points of the divisor $D$, hence the basis invariant $p_i$ ranges over $H^0(d_iD)$; this proves $2^\circ$. Let $h^0(d_iD)$ stay for the dimension of the last space (for the exceptional case of the series $D_l$ this is $p_l^2$ which ranges over $H^0(d_rD)$; it does not affect the dimension).

The dimension of the space of all spectral invariants (coinciding with the number of independent integrals, and also with the dimension of the space of spectral curves) is equal to
\[
   N=\sum_{i=1}^l h^0(d_iD)
\]
Since $\g$ is semi-simple, the degrees of its basis invariant polynomials are not less than 2, hence for $\deg D>g-1$ the divisors $d_iD$ are non-special ($\deg d_iD> 2(g-1)$) which enables one to compute $N$ by means the Riemann--Roch theorem. For $D\in\K$ ($\deg D=2(g-1)$) the Riemann--Roch theorem gives
\[
     h^0(d_iD)=d_i\deg D-g+1=(2d_i-1)(g-1),
\]
hence
\[
   N=\sum_{i=1}^l (2d_i-1)(g-1)=\dim\g\cdot(g-1)
\]
due to the identity $\sum_{i=1}^l (2d_i-1)=\dim\g$ valid for an arbitrary complex semisimple Lie algebra. This calculation became standard since \cite{Hitchin}, see  \cite{Shein_UMN2016} for its version in the present set-up.
\end{proof}
\begin{remark}
Vise verse, for the reductive Lie algebras, for example, for $\gl(n)$, the coefficient at $\l^{n-1}$ (i.e. the trace of $L$, the 1st degree invariant)  does not vanish in general, hence $d_1=1$, and the first summand in the expression for $N$ is exactly $h^0(D)$, but the divisor $D\in\K$ is special.
\end{remark}
\subsection{Dimension of the phase space}\label{SS:phas}
Here we show that the dimension of the phase space is equal to $2N=2\dim\g\cdot(g-1)$. We will do it here for $\g=\sln(n)$ and refer to \cite{Shein_UMN2016} for the remainder of the classical simple Lie algebras, and to \cite{Sh_TMPh} for $\g=G_2$.

Thus we set $\g=\sln(n)$ until the end of this section. Then the local conditions satisfied by $L$ are given in the \refEx{Ex1}. We denote the space of Lax operators with fixed Tyurin parameters and the divisor $D$ by $\L^D_{\a,\Gamma}$ where $\a=\{\a_\ga\ |\ \ga\in\Gamma\}$. The dimension of  $\L^D_{\a,\Gamma}$ is equal to $\dim\g\cdot(\deg D-g+1)$ \cite{Shein_UMN2016}. For $D\in\K$
\[
    \dim \L^\K_{\a,\ga}=\dim\g\cdot(g-1).
\]
To obtain the full dimension of the phase space we must add the number of Tyurin parameters and subtract the gauge degree of freedom. For $\g=\sln(n)$ the number of the Tyurin parameters at one of the points $\ga$ is equal to $n$, the number of the points is a product $(\rank\g)g=(n-1)g$. Thus the total number of Tyurin parameters modulo gauge freedom is equal to  $n(n-1)g+(n-1)g-(n^2-1)=(n^2-1)(g-1)=\dim\g\cdot(g-1)$, where $n(n-1)g$ stays for the total number of the parameters $\a$, $(n-1)g$ is the number of the points $\ga$, and $n^2-1=\dim\g$.

All together, we have for the dimension of the phase space (being denoted by $\L^\K$)
\[
   \dim\L^\K=2\dim\g\cdot(g-1)=2N
\]
where $N$ is the number of integrals.

In \refS{separ} below we will give another description of the phase space of hyperelliptic Hitchin systems based on the description of spectral curves (\refSS{dx/y}).
\subsection{Basis in the space of spectral invariants for hyperelliptic curves (the case $\varpi=dx/y$)}\label{SS:dx/y}
For classical Lie algebras the spectral curve is given by the equation of the form
\begin{equation}\label{E:urav}
   R(x,y,\l)=\l^n+\sum_{i=1}^l r_i(x,y)\l^{n-d_i}=0
\end{equation}
where $n$ is the dimension of the \emph{standard} (in other terminology \emph{vector}) representation of the Lie algebra $\g$, $l=\rank\g$, and $r_i$ ($i=1,\ldots,r$) are meromorphic functions on $\Sigma$. Here, $r_i(z)=\chi_i(L(z))$ where $\chi_i$ ($i=1,\ldots,r$) are basis invariant polynomials of $\g$, and $\deg \chi_i=d_i$. The series $D_l$ is exceptional with this respect: the summand of degree 0 in $\l$ in \refE{urav} is of the form $p_l(z)=\chi_l(L(z))^2$ (i.e. the square of the basis invariant; indeed the basis invariant is nothing but the Pfaffian in this case). Jumping ahead, we notice that this is the reason why the equations for Hamiltonians in the method of separation of variables become non-linear for the $D_l$ series.

The differential $\varpi=dx/y$ on the curve \refE{curve} is holomorphic and has a zero of the multiplicity $2(g-1)$ at $\infty$: $D=(\varpi)=2(g-1)\cdot\infty$. The basis spectral invariants of order $d_i$ run over the space $H^0(d_iD)$. We want to exhibit a base in this space.
\begin{proposition}\label{P:f_base}
The functions $1,x,\ldots,x^{d_i(g-1)}$, and $y,yx,\ldots,yx^{(d_i-1)(g-1)-2}$ form a base in the space $H^0(d_iD)$ for $D=2(g-1)\cdot\infty$.
\end{proposition}
\begin{proof}
The function $x$ has a pole of order 2 at the infinity, hence $x^k\in H^0(d_iD)$ if, and only if $0\le 2k\le 2d_i(g-1)$, which implies $0\le k\le d_i(g-1)$. These functions are linearly independent, and their number is equal to $d_i(g-1)+1$.

The function $y$ has a pole of order $2g+1$ at the infinity, hence $yx^s\in H^0(d_iD)$ if, and only if $s\ge 0$ and $2g+1+2s\le 2d_i(g-1)$. The second inequality implies $2s\le 2d_i(g-1)-2g-1$. Because the left hand side of the inequality is even while the right hand side is odd, we obtain $2s\le 2d_i(g-1)-2g-2$. All together, $0\le s\le (d_i-1)(g-1)-2$. The number of such functions is equal to $0\le s\le (d_i-1)(g-1)-1$.

The total number of the above functions is equal to $(2d_i-1)(g-1)$ which coincides with the above computed dimension of the space $H^0(d_iD)$. They  are linearly independent for the reason their orders at the infinite point are different. In particular, the functions in the first set have even orders, while the orders of the functions in the second set are odd. All together, they form a base in the space $H^0(d_iD)$.
\end{proof}
The expansion of the coefficients of the spectral curve over the just found out base has the form
\begin{equation}\label{E:coeff_spc}
    r_i(x,y)=\sum_{k=0}^{d_i(g-1)} H^{(0)}_{ik} x^k +\sum_{s=0}^{(d_i-1)(g-1)-2} H^{(1)}_{is}yx^s
\end{equation}
where $H^{(0)}_{ik}$, $H^{(1)}_{is}$ are independent integrals of the Hitchin system.

\subsection{Example: the form of a spectral curve for $\g=\sln(2)$ and $\varpi=dx/y$}\label{SS:r2g2}
In this case, there is only one basis spectral invariant $r_2$ ($i=2$, $d_i=2$):
\begin{equation}\label{E:r2_sl2}
   r_2(x,y)=\sum_{k=0}^{2(g-1)} H^{(0)}_{k} x^k +\sum_{s=0}^{g-3} H^{(1)}_{s}yx^s
\end{equation}
(in particular, for $g=2$ the second sum is absent, see also \cite{PBor}). The spectral curve has the form
\begin{equation}\label{E:sp_c_sl2}
      \l^2+r_2(x,y)=0.
\end{equation}
\subsection{Basis of holomorphic differentials on a spectral curve}
\begin{proposition}\label{P:h_diff}
If the base and the spectral curves are non-degenerate and have only simple, mutually different branch points in the finite domain (as branch coverings of the $x$-line), then a base in the space of holomorphic differentials on the spectral curve is given by $\frac{x^k\l^{n-d_i}dx}{R_\l'(x,y,\l)\, y }$ $(0\le k\le d_i(g-1))$, and $\frac{x^s\l^{n-d_i}dx}{R_\l'(x,y,\l)}$ $(0\le k\le {(d_i-1)(g-1)-2})$, $i=1,\ldots,l$.
\end{proposition}
\begin{proof}
It is a standard fact that for the non-degenerate plane curves the differential $\frac{dx}{R_\l'(x,y,\l)}$ is holomorphic at simple branch points in the finite domain. It is also true in our case. Indeed, in the neighborhood of such point $R'_\l=0$, $R'_x\ne 0$ hence $\l$ can be chosen as a local parameter, and $\frac{dx}{R_\l'(x,y,\l)}=\frac{x'_\l d\l}{R_\l'(x,y,\l)}$. By $R'_xx'_\l+R'_\l=0$ we obtain $\frac{dx}{R_\l'}=-\frac{d\l}{R_x'}$. Since the $dx/y$ is also holomorphic, and branch points of the two curves are mutually different, we obtain that $\frac{dx}{R_\l'\, y}$ is holomorphic too. The nominators $x^k\l^{n-d_i}$ ($x^s\l^{n-d_i}$, respectively) obviously have no pole in the finite domain.

At the infinity, we observe that $x\sim z^{-2}$, $y\sim z^{-2g-1}$ (which comes from the base curve), and $\l\sim z^{-2(g-1)}$ where $z$ is a local parameter (we use notation $y\sim z^{-2g-1}$ etc., in the same sense as $y=O(z^{-2g-1})$). The last comes because the highest terms in \refE{coeff_spc}, under this assumption, are of the same order $n$ (in $\l$) independently of $i$.

Obviously, for a given $i$, the differentials listed in the \refP{h_diff} are products of the functions listed in \refP{f_base} by the differential $\frac{\l^{n-d_i}dx}{R'_\l\, y}$. Observe that the functions in \refP{f_base} are found out from the condition equivalent to the requirement that their products by $\l^{-d_i}$ are holomorphic at $\infty$.  For example, $x^k\l^{-d_i}\sim z^{-2k}z^{2(g-1)}$. For $k\le d_i(g-1)$ (as stated in \refP{f_base}) $x^k\l^{-d_i}$ is holomorphic at $\infty$.

It remains only prove that the differential $\frac{\l^ndx}{R'_\l\, y}$ is holomorphic at $\infty$. Indeed, $R'_\l~\sim ~\l^{n-1}$, hence $\frac{\l^ndx}{R'_\l\, y}\sim\frac{\l dx}{y}\sim \frac{z^{-2(g-1)}z^{-3}dz}{z^{-(2g+1)}}\sim dz$.

The linear independence of the differentials in question descends to that of the functions in \refP{f_base}.
\end{proof}
\section{Separation of variables for hyperelliptic Hitchin systems}\label{S:separ}
\subsection{Phase space. Symplectic form and Poisson structure}
A spectral curve of the form \refE{coeff_spc} can be given by $(\dim\g)(g-1)$ points it passes through. Denote these points by $(x_i,y_i,\l_i)$ ($i=1,\ldots,(\dim\g)(g-1)$) where $y_i^2=P_{2g+1}(x_i)$, and all three are related by the equations \refE{urav}, \refE{coeff_spc}.

Define a two-form by
\begin{equation}\label{E:sympl_str}
   \w=\sum_{i=1}^{(\dim\g)(g-1)} d\l_i\wedge\frac{dx_i}{y_i}.
\end{equation}
According to \cite{Kr_Lax}, in particular to Theorem 4.3 therein, the symplectic structure given by $\w$ indeed is that of the Hitchin system.
\begin{remark}
To retrieve the Hitchin system from our data one has to consider the points $(x_i,y_i,\l_i)$ as poles of the eigenfunction of the Lax operator of Hitchin system, and use the technique of the inverse scattering method as described in \cite{Kr_Lax}. Indeed \refE{sympl_str} is nothing but 	adaptation of the Theorem 4.3 \cite{Kr_Lax} for the case of hyperelliptic (base) curve. In particular, the second wedge co-multiplier in \refE{sympl_str} is nothing but the differential $\frac{dx}{y}$ fixed in the definition of the Lax operator  (also in the definition of the symplectic structure in \cite{Kr_Lax}).
\end{remark}
The phase space of a Hitchin system on a hyperelliptic curve $\Sigma$ is formed by non-ordered sets $\{ (x_i,y_i,\l_i)|i=1,\ldots, (\dim\g)(g-1)\}$, with symplectic structure given by \refE{sympl_str}.

The corresponding Poisson structure is given by
\begin{equation}\label{E:poiss_str}
   \{ \l_i,x_j \}=\d_{ij}y_i.
\end{equation}
\subsection{Angle coordinates}
The subvariety of triples $\{ (x_i,y_i,\l_i)|i=1,\ldots, (\dim\g)(g-1)\}$ giving the same spectral curve $C$ is equal to $S^{(\dim\g)(g-1)}C$ (the symmetric power of $C$, a triple $(x_i,y_i,\l_i)$ belongs to the $i$th copy of $C$).  The Abel transformation maps it to the Jacobian of $C$ (we will denote it by $Jac(C)$). By \emph{angle coordinates} we mean the coordinates on the Jacobian:
\begin{equation}\label{E:angle_coord1}
\phi_{jk}^{(0)}=\sum_{i=1}^{(\dim\g)(g-1)} \int^{\ga_i}\frac{x^k\l^{n-d_j}dx}{R_\l'(x,y,\l)\, y },\ 0\le k\le d_j(g-1);
\end{equation}
\begin{equation}\label{E:angle_coord2}
\phi_{js}^{(1)}=\sum_{i=1}^{(\dim\g)(g-1)} \int^{\ga_i}\frac{x^s\l^{n-d_j}dx}{R_\l'(x,y,\l)},\ 0\le s\le {(d_j-1)(g-1)-2}
\end{equation}
for $j=1,\ldots,l$, where $\ga_i=(x_i,y_i,\l_i)$ (cf. \cite[Eq. (4.61)]{Kr_Lax}).

\subsection{Action coordinates (for $\g$ of the $A_l$, $B_l$, $C_l$ type)}\label{SS:act_coord}
For these series the Hamiltonians of the Hitchin system can be expressed via $(x_i,y_i,\l_i)$ as a solution to the system of linear equations $R(x_i,y_i,\l_i)=0$ where $R(x,y,\l,H)=0$ is the equation of the spectral curve given by \refE{urav}, \refE{coeff_spc} ($H=\{ H^{(0)}_{jk}, H^{(1)}_{js}\ |\ j=1\ldots,l;\ k=0,1,\ldots,(2d_j-1)(g-1),\ s=0,1,\ldots, (d_j-1)(g-1)-2 \}$). In particular, for $\g=\sln(2)$ the system of equations is of the form
\begin{equation}\label{E:sep_rel}
  \l_i^2+\sum_{k=0}^{2(g-1)} H^{(0)}_{k} x_i^k +\sum_{s=0}^{g-3} H^{(1)}_{s}y_ix_i^s=0,\ i=1,\ldots, 3(g-1).
\end{equation}
Thus $H^{(0)}_{k}=D^{(0)}_{k}/D$, $H^{(1)}_{s}=D^{(1)}_{s}/D$ where
\[
D=\begin{vmatrix}
     1 & \ldots & x_1^{2(g-1)}   & y_1    & \ldots & y_1x_1^{g-3}  \\
\vdots &        & \vdots         & \vdots &        & \vdots \\
    1 & \ldots & x_{3(g-1)}^{2(g-1)}   & y_{3(g-1)} & \ldots & y_{3(g-1)}x_{{3(g-1)}}^{g-3}
             \end{vmatrix},
\]
$D^{(0)}_{k}$ is obtained by replacing $x_i^k$ by $(-\l_i^2)$ in the $k$th column in $D$, $D^{(1)}_{s}$ is obtained by replacing $y_ix_i^k$ by $(-\l_i^2)$ in the $(2g-1+k)$th column in $D$.

In the next section, we prove that the $H$- and $\phi$-coordinates introduced in the last two sections indeed are Darboux coordinates for the symplectic structure \refE{sympl_str}.

We are not in position to claim the same for the series $D_l$ because in this case the system of equations on the Hamiltonians is quadratic.
\subsection{Darboux property}
The Darboux property for the coordinates $(H^{(0)}_{jk}, \phi^{(0)}_{jk})$ ($k=0,1,\ldots,(2d_j-1)(g-1)$), $(H^{(1)}_{js}, \phi^{(1)}_{js})$ ($s=0,1,\ldots, (d_j-1)(g-1)-2$), $j=1\ldots,l$ is an immediate corollary of the following lemma.
\begin{lemma}[\cite{DT}]\label{L:LDT}
Consider a curve given by the equation
\[
    R(\l,\mu)=R_0(\l,\mu)+\sum_{j=1}^n H_jR_j(\l,\mu)=0,
\]
and the space of sets $(\l_1,\mu_1,\ldots,\l_n,\mu_n)$ with the Poisson bracket $\{\l_i,\mu_j\}=\d_{ij}f(\l_i,\mu_i)$ where $f$ is a smooth function.
Let $H_1,\ldots,H_n$ be the solution to the linear system of equations $R(\l_i,\mu_i)=0$, and $\phi_1,\ldots,\phi_n$ are given by
\begin{equation}\label{E:TD_angles}
    \phi_j=\sum_{k=1}^n\int^{\ga_k} \frac{R_j(\l,\mu)d\mu}{\partial_\l R(\l,\mu)f(\l,\mu)}
\end{equation}
where $\ga_k=(\l_k,\mu_k)$. Then $\{ H_i,\phi_j\}=\d_{ij}$.
\end{lemma}
Indeed, as it was observed in course of the proof of \refP{h_diff}, in our case the integrands in \refE{TD_angles} are exactly the same as the basis differential listed in \refP{h_diff}, with $f=y$.  Also $H_j$ in the \refL{LDT}, and in \refSS{act_coord} are obtained in the same way (going back to \cite{BT2}).

\section{Bases in the spaces of spectral invariants and Krichever--Novikov functions}\label{S:KN}

It is not possible to exhibit a basis like  that given in \refP{f_base} for a generic choice of the holomorphic differential $\varpi$. It is nevertheless possible to give a base in another form using Krichever--Novikov basis functions. Here we will illustrate this in the case $\varpi=x^{g-1}dx/y$, inspite in this case the base in terms of monomials in $x$, $y$ still does exist: for $g=2$, $\g=\sln(2)$ it is exhibited by D.Talalaev\footnote{Personal communication} (the exhibited base turned out to consist of inverse ones to the functions of the base exhibited in \refSS{r2g2}, and in \cite{PBor} for $g=2$, $\g=\sln(2)$, and $\varpi=dx/y$).

\subsection{A base in the space of the spectral invariants for $\varpi=x^{g-1}dx/y$}

We set $\varpi=x^{g-1}dx/y$ below (this is the only differential of the form $x^kdx/y$, $0\le k\le g-1$ which has no zero at $x=y=\infty$). The differential $\varpi$ is holomorphic and has two zeroes $0_1$ and $0_2$ of multiplicity $g-1$ over $x=0$ (we assume that $x=0$ is not a branching point of $\Sigma$). Thus, $D=(g-1)0_1+(g-1)0_2$.

The basis spectral invariants of an order $d_i$ run over the space $H^0(d_iD)$. The dimension of this space is equal $(2d_i-1)(g-1)$, as computed above. We will look for a basis in this space among the basis Krichever--Novikov functions corresponding to two incoming points $0_1$, $0_2$, and one outgoing point $Q$ to be taken arbitrarily in a generic position. Such functions are enumerated by the pairs $(n,r)$, $n\in\Z$, $r=1,2$, are denoted by $A_{n,1}$, $A_{n,2}$, and are given by the following relations \cite[relation (1.3.3)]{Sh_DGr}:
\begin{align*}
  & \ord_{O_1}A_{n,1}=\ord_{O_2}A_{n,2}=n,      \\
  &  \ord_{O_1}A_{n,2}=\ord_{O_2}A_{n,1}=n+1,     \\
  &  \ord_Q A_{n,1}=\ord_Q A_{n,2}=-2(n+1)-(g-1).
\end{align*}
For a generic curve the points $0_1$ and $0_2$ are not Weierstra\ss\ points, and a Krichever--Novikov-type base exists. Moreover, its elements have known expressions in terms of $\theta$-functions \cite{Schlich_DGr}. In particular, the following statement holds.
\begin{theorem}[\cite{Schlich_DGr}]
If $m>0$ or $m<-\left[ \frac{g-2}{2}\right]-2$, $r=1,2$ then the functions $A_{m,r}$ are uniquely defined.
\end{theorem}
$A_{m,r}\in H^0(d_iD)$ if, and only if $\ord_Q A_{m,r}\ge 0$ ($r=1,2$). Hence $-2(n+1)-(g-1)\ge 0$ and
\begin{equation}\label{E:nrange1}
     -d_i(g-1)\le n\le -\frac{g+1}{2}
\end{equation}
If $g$ is odd then the upper bound is integer here, and the total number of functions satisfying the restriction \refE{nrange1}, is equal to $2\left(  -\frac{g+1}{2}+ d_i(g-1)+1  \right)=(2d_i-1)(g-1)$  (which coincides with the dimension $H^0(d_iD)$). If $g$ is even then one function is missing, but it is compensated by adding the identical unit to the set of the basis functions.

Thus a general form of the spectral curve of a hyperelliptic genus $g$ Hitchin system with a semisimple Lie algebra $\g$ is as follows:
\begin{equation}\label{E:spec_gen}
   R(\l,P)= \l^n+\sum_{j=1}^l \l^{n-d_j}r_j(P)    ,
\end{equation}
where
\[
   r_j=\sum_{k=0}^{(2d_j-1)(g-1)-1} H_{jk}A_k ,
\]
$H_{ij}$ are the integrals of the system.
\begin{example}
Let $\varpi=x^{g-1}dx/y$ and $\g=\sln(2)$.
For the second order spectral invariants $(r_2)\ge -2D$, hence $r_2\in H^0(2D)$. As it was computed above, $\dim H^0(2D)=3(g-1)$.  For the above functions $A_{n,r}$ ($r=1,2$), $A_{n,r}\in H^0(2D)$ if, and only if $\ord_Q A_{n,r}\ge 0$, that is $-2(n+1)-(g-1)\ge 0$, and
\begin{equation}\label{E:nrange}
     -2(g-1)\le n\le -\frac{g+1}{2}
\end{equation}
The number of the basis elements satisfying this condition is equal to $2(-\frac{g+1}{2}+2(g-1)+1)=3(g-1)$ which coincides with the dimension of the space $H^0(2D)$. We introduce the following enumeration of those elements: $A_{n,r}=A_j$ where $j=2(n+2(g-1))+r$. Then $j=1,\ldots, 3(g-1)$. The elements $A_j$, $j=1,\ldots, 3(g-1)$ form a base in $H^0(2D)$. If $g$ is odd we obtain one basis element less that it is required. The missing element can be fixed using the Krichever--Noikov duality between functions and 1-forms \cite{Schlich_DGr, Sh_DGr}.
\end{example}
\subsection{Darboux coordinates}
The Hamiltonians can be expressed from the system of linear equations
\begin{equation}\label{E:spec_KN}
    R(\ldots,H_{jk},\ldots\l_i,P_i)=0,\ i=1,\ldots,N
\end{equation}
by the Kramer's rule, where $R$ is defined by the relation \refE{spec_gen}, $j=1,\ldots,n$, $k=0,\ldots,N$, $N=(\dim\g)(g-1)-1$, $(P_i, l_i)$ ($i=1,\ldots,N$) are points of the spectral curve.
To point out the conjugated coordinates, we again make use of the relation \refE{TD_angles}. In the present case the curve is of the form \refE{spec_gen}, hence introducing the local coordinate from the relation $dz=\varpi$ we obtain $R_{jk}(\l,z)=\l^{n-d_j}H_{jk}A_k(z)$,
\[
   \partial_\l R(\l,P)= n\l^{n-1}+\sum_{j=1}^l (n-d_j)\l^{n-d_j-1}\sum_{k=0}^{(2d_j-1)(g-1)-1} H_{jk}A_k(z)    ,
\]
and finally we obtain
\[
   \phi_{jk}=\sum_k\int^{z_{jk}}\frac{\l^{n-d_j}H_{jk}A_k(z)}{n\l^{n-1}+\sum_{j=1}^l (n-d_j)\l^{n-d_j-1}\sum_{k=0}^{(2d_j-1)(g-1)-1} H_{jk}A_k(z)}dz.
\]
Since we have not check any holomorphicity of the integrands, we do not claim that the coordinates $H_{jk},\phi_{jk}$ indeed are the action--angle coordinates in this case (though we would expect that).
\begin{example}
The system \refE{spec_KN} is especially simple for $\g=\sln(2)$, it is of the form
\[
   \l_i^2-\sum_{k=1}^NA_{k-1}(P_i)H_k=0,\quad i=1,\ldots,N
\]
where $N=3(g-1)$. From that, we have
\[
   H_j=\frac{\begin{vmatrix}
    A_0(P_1) & \ldots & A_{j-1}(P_1)  & \l_1^2 & A_{j+1}(P_1) & \ldots & A_{N-1}(P_1)  \\
          \vdots &        & \vdots & \vdots & \vdots &        & \vdots \\
    A_0(P_N) & \ldots & A_{j-1}(P_N)  & \l_N^2 & A_{j+1}(P_N) & \ldots & A_{N-1}(P_N) \\
             \end{vmatrix}
   }{\det(A_{k-1}(P_i))_{i,k=1,\ldots,N}}.
\]
Also, $\partial_\l R(\l,P)=2\l$, and $R_k(\l,z)=A_k(z)$, hence
\[
   \phi_{k}=\frac{1}{2}\sum_k\int^{z_{k}}\frac{A_k(z)}{\l}dz.
\]
\end{example}
\section{Concluding remarks}\label{S:concluding}
What is said in \refS{descr} and \refS{separ}, enables one to define a Hitchin system on a genus $g$ hyperelliptic curve, with a Lie algebra $\g$, as a system with $(\dim\g)(g-1)$ degrees of freedom, phase space formed by the triples $(x_i,y_i,\l_i)$ satisfying the relations \refE{curve}, \refE{urav}, \refE{coeff_spc} ($i=1,\ldots,(\dim\g)(g-1)$), having a symplectic (Poisson) structure as defined in \refE{sympl_str} (resp., in \refE{poiss_str}), whose Hamiltonians are defined from the system of equations
\[
   R(x_i,y_i,\l_i,H) =0,\quad i=1,\ldots,(\dim\g)(g-1)
\]
where $R(x,y,\l,H)$ is defined by \refE{urav} and \refE{coeff_spc}, $H$ is a set of independent Hamiltonians. Integrability of such a system can be proved independently of any Lax representation, Hamiltonian reduction, and other methods conventionally used for this purpose. It follows from the following elementary statement.
\begin{proposition}[\cite{Aint_sys}]\label{P:comm}
Let $H_j=H_j(z_1,\ldots,z_n,\l_1,\ldots,\l_n)$ ($j=1,\ldots,n$) be functions defined as a solution to the system of equations
\begin{equation}\label{E:non_sys}
F_i(H_1,\ldots,H_n,z_i,\l_i)=0,\quad i=1,\ldots,n
\end{equation}
where $F_i$ are given smooth function of complex or real  variables. Then (under certain natural genericity requirements) $H_1,\ldots,H_n$ commute with respect to any Poisson bracket of the form
$
\{ f,g\}=\sum\limits_{j=1}^np_j\left(\frac{\partial f}{\partial z_j}\frac{\partial g}{\partial \l_j}-\frac{\partial g}{\partial z_j}\frac{\partial f}{\partial \l_j}\right)
$,
where $p_j=p_j(z_j,\l_j)$ are smooth functions in only one pair of variables
\emph{(in our case $x_i$ plays the role of $z_i$, for all $i$)}.
\end{proposition}

Another remark is as follows. A choice of the holomorphic differential $\varpi$ on $\Sigma$ in the definition of the Lax operator (see \refSS{Lax} for details) is a kind of gauge freedom. The results of \refS{descr}, \refS{separ} correspond to the choice $\varpi=dx/y$. This is a very special choice in sense that the divisor of $\varpi$ is supported at one point, namely at $\infty$, and this is a Weierstra\ss\ point. It is not obvious that a similar technique would turn out to be successfull for a generic choice of $\varpi$. However, there is another (complementary) option in this case, namely, to express the coefficients of spectral curves via certain Krichever--Novikov basis functions (\refS{KN}). It was a basic approach in the early version of this work. However, in the present version we focus on the first approach due to discussion with D.Talalaev and S.P.Novikov at the S.P.Novikov's seminar.

\bibliographystyle{amsalpha}

\end{document}